%% file: main.tex
\documentclass[submission,copyright,creativecommons]{eptcs}
\providecommand{\event}{ICE 2017}

\usepackage{amssymb}
\usepackage{amsmath}
\usepackage{amsthm}
\usepackage{stmaryrd}
\usepackage{enumerate}
\usepackage{proof-dashed}
\usepackage{multirow}
\usepackage[all]{xy}
\usepackage{tikz}
\usetikzlibrary{arrows}
\usetikzlibrary{positioning}
\usepackage{listings}
\usepackage{comment}
\usepackage{calrsfs}
\usepackage{todonotes}

\newcommand{\journal}[1]{}
\newcommand{\conf}[1]{#1}

\input{macros}

\newtheorem{theorem}{Theorem}
\newtheorem{corollary}{Corollary}
\newtheorem{definition}{Definition}
\newtheorem{example}{Example}
\newtheorem{remark}{Remark}

\title{On Asynchrony and Choreographies}
\author{Lu\'is Cruz-Filipe
\institute{Department of Mathematics and Computer Science\\
University of Southern Denmark\\
Odense, Denmark}
\email{lcfilipe@imada.sdu.dk}
\and
Fabrizio Montesi
\institute{Department of Mathematics and Computer Science\\
University of Southern Denmark\\
Odense, Denmark}
\email{fmontesi@imada.sdu.dk}
}
\def\titlerunning{On Asynchrony and Choreographies}
\def\authorrunning{L. Cruz-Filipe \& F. Montesi}

\begin{document}

\maketitle

\begin{abstract}
\input{abstract}
\end{abstract}

\section{Introduction}
\label{sec:intro}
\input{intro}

\section{Minimal Choreographies and Stateful Processes}
\label{sec:cc-sp}
\input{cc-sp}

\section{Asynchronous Choreographies}
\label{sec:async}
\input{async}

\section{Asynchronous Processes}
\label{sec:asp}
\input{asp}

\section{Related Work}
\label{sec:related}
\input{related}

\section{Conclusions and Future Work}
\label{sec:concl}
\input{conclusions}

\smallpar{Acknowledgements}
This work was supported by CRC (Choreographies for Reliable and efficient Communication software), grant no.\ DFF--4005-00304 from the Danish Council for Independent Research, by grant DFF-1323-00247 from the Danish Council for Independent Research, Natural Sciences, and by the Open Data Framework project at the University of Southern Denmark.


\bibliographystyle{eptcs}
\bibliography{biblio}

\end{document}

%% file: macros.tex
\newcommand{\m}[1]{\ensuremath{\mathsf{#1}}}
\newcommand{\pid}[1]{\m{#1}}

\newcommand{\ctxc}[2][]{{\mathcal C}_{#1}[#2]}
\newcommand{\ctx}[1][\bullet]{\ctxc{#1}}

\newcommand{\emptyN}{\nil} 

\newcommand{\code}[1]{\texttt{#1}}
\newcommand{\nil}{\mathbf{0}}
\newcommand{\com}[2]{#1\;\code{-\hspace{-0.3mm}>}\;#2}
\newcommand{\gencom}{\com{\pid p.e}{\pid q}}
\newcommand{\sel}[3]{\com{#1}{#2 [#3]}}
\newcommand{\gensel}{\sel{\pid p}{\pid q}{\ell}}

\newcommand{\toa}{\to_a}

\renewcommand{\merge}{\sqcup}
\newcommand{\cond}[3]{\m{if}\, #1 \, \m{then} \, #2 \, \m{else} \, #3}
\newcommand{\gencond}{\cond{\pid p.e}{C_1}{C_2}}

\newcommand{\rec}[3]{\m{def} \, #1  =  #2 \, \m{in} \, #3}
\newcommand{\genrec}{\rec{X}{C_2}{C_1}}
\newcommand{\genrecb}{\rec{X}{B_2}{B_1}}

\newcommand{\pn}{\m{pn}}

\newcommand{\gencall}{X}

\definecolor{light-gray}{gray}{0.928}

\newcommand{\epp}[2]{[\![#1]\!]_{#2}}

\newcommand{\rname}[2]{\ensuremath{\left\lfloor\mbox{{#1}$|${#2}}\right\rceil}}
\newcommand{\asend}[2]{{#1}!{#2}}
\newcommand{\arecv}[1]{#1?}
\newcommand{\actor}[4][]{{#2} \triangleright^{#1}_{#3} {#4}}
\newcommand{\parp}{\mathrel{\boldsymbol{|}}}
\newcommand{\asel}[2]{{#1}\oplus#2}
\newcommand{\abranch}[2]{{#1}\&{#2}}
\newcommand{\precongr}{\preceq}

\newcommand{\defeq}{\stackrel{\Delta}{=}}

\newcommand{\rhop}{\rho_{\pid p}}
\newcommand{\rhoq}{\rho_{\pid q}}

\newcommand{\smallpar}[1]{\smallskip\noindent \textbf{\textit{#1.}}}

\newcommand{\upd}[3]{\m{upd}(#1,#2,#3)}
\newcommand{\update}[2]{\m{upd}(#1,#2)}
\newcommand{\eval}[2]{\mathrel{\downarrow_{#1(#2)}}}
\newcommand{\evals}[1]{\mathrel{\downarrow_{#1}}}
\newcommand{\hd}[2]{\m{next}_{#2}(#1)}

\newcommand{\assend}[4][]{{#2}.{#3}\;\code{-\hspace{-0.3mm}>}_{#1}\;{#4}}
\newcommand{\asrecv}[3][]{{#2}\;_{#1}\code{-\hspace{-0.3mm}>}\;{#3}}
\newcommand{\genasend}{\assend{\pid p}{e}{x}}
\newcommand{\genarecv}{\asrecv{\maybe v}{\pid q}}
\newcommand{\maybe}[1]{\hat{#1}}
\newcommand{\stepp}[2]{(\!|{#1}|\!)_{#2}}

%% file: abstract.tex
Choreographic Programming is a paradigm for the development of concurrent software, where deadlocks 
are prevented syntactically.
However, choreography languages are typically synchronous, whereas many real-world systems have asynchronous 
communications.
Previous attempts at enriching choreographies with asynchrony rely on \emph{ad-hoc} constructions, whose 
adequacy is only argued informally.
In this work, we formalise the properties that an asynchronous semantics for choreographies should have:
messages can be sent without the intended receiver being ready, and all sent messages are eventually received.
We explore how out-of-order execution, used in choreographies for modelling 
concurrency, can be exploited to endow choreographies with an 
asynchronous semantics. Our approach satisfies the properties we identified.
We show how our development yields a pleasant correspondence with FIFO-based asynchronous messaging, modelled in a 
process calculus, and discuss how it can be adopted in more complex choreography models.

%% file: intro.tex
Choreographic Programming~\cite{M13:phd} is a paradigm for developing concurrent software, where an
``Alice and Bob'' notation is used to prevent mismatched I/O actions syntactically.
An EndPoint Projection (EPP) can then be used to synthesise correct-by-construction process 
implementations~\cite{CM13,CM17:facs,QZCY07}.
Choreographies are used in different settings, including standards~\cite{BPMN,wscdl}, 
languages~\cite{chor:website,HMBCY11,pi4soa,savara:website}, specification models~\cite{CHY12,CM13,LGMZ08}, and design 
tools~\cite{BPMN,pi4soa,savara:website,wscdl}.

The key to preventing mismatched I/O actions in choreographies is that interactions between two (or more) processes are 
specified atomically, using terms such as $\gencom;C$ -- read ``process $\pid p$ sends the evaluation of expression $e$ 
to process $\pid q$, and then we proceed as the choreography $C$''.
Giving a semantics to such terms is relatively easy if we assume that communications are synchronous: we can just 
reduce $\gencom;C$ to $C$ in a single step (and update $\pid q$'s state with the received value, but this is orthogonal 
to this discussion).
For this reason, most research on choreographic programming focused on systems with a synchronous communications 
semantics.

However, many real-world systems use asynchronous communications.
This motivated the introduction of an ad-hoc reduction rule for modelling asynchrony in choreographies 
(Rule~\rname{C}{Async} in~\cite{CM13}).
As an example, consider the choreography $\com{\pid p.1}{\pid q};\com{\pid p.2}{\pid r}$ (where $1$ and $2$ are just 
constants). The special rule would allow 
for consuming the second communication immediately, thus reducing the choreography to $\com{\pid p.1}{\pid 
q}$. In general, roughly, this rule allows a choreography of the form $\gencom;C$ to execute an action in 
$C$ if this action involves $\pid p$ but not $\pid q$ (sends are non-blocking, 
receives are blocking).
This approach was later adopted in other works (see Section~\ref{sec:related}). Unfortunately, it also comes with a serious 
problem: it yields an unintuitive semantics, since a choreography can now reduce to a state that would normally not be 
reachable in the real world. In our example, specifically, in the real world $\pid p$ would have to send its first 
message to $\pid q$ before it could proceed to sending its other message to $\pid r$. This information is lost in the 
choreography reduction, where it appears that $\pid p$ can just send its messages in any order.
In~\cite{CM13}, this also translates to a misalignment between the structures of choreographies and their process 
implementations generated by EPP, since the latter use a standard asynchronous semantics with 
message buffers; see Section~\ref{sec:related} for a detailed discussion of this aspect.
Previous work~\cite{DY13} uses intermediate runtime terms in choreographies to represent asynchronous messages in 
transit, in an attempt at overcoming this problem. However, the adequacy of this approach has never been formally 
demonstrated.

In this paper, we are interested in studying asynchrony for choreographies in a systematic way. Thus, we 
first analyse the properties that an asynchronous choreography semantics should have (assuming standard FIFO duplex 
channels between each pair of processes) -- messages can be sent without the intended receiver being ready, and all sent 
messages are eventually received --
and afterwards formulate them 
precisely in a representative choreography language. Our study leads naturally to the construction of a
new choreography model that supports asynchronous communications, by capitalising on the characteristic feature of 
out-of-order execution found in choreographic programming.
We formally establish the adequacy of our asynchronous model, by proving that it respects the formal definitions of our 
properties.
Then, we define an EPP from our new model to an asynchronous process calculus.
Thanks to the accurate asynchronous semantics of our choreography model, we prove that the code generated 
by our EPP and the originating choreography are lockstep operationally equivalent.
As a corollary, our generated processes are deadlock-free by construction. Our development also has the pleasant 
property that programmers do not need to reason about asynchrony: they can just program thinking in the usual terms of 
synchronous communications, and assume that adopting asynchronous communications will not lead to errors.
We conclude by discussing how our construction can be systematically extended to more complex choreography models.

\paragraph{Contribution.}
The contribution of this article is threefold. 
First, we give an abstract characterisation of asynchronous semantics for choreography languages, which is formalised
for a minimal choreography calculus.
Secondly, we propose an asynchronous semantics for this minimal language, show that it is an instance of our 
characterisation, and discuss how it can be applied to other choreography calculi.
Finally, we prove a lockstep operational correspondence between choreographies and their process implementations, when 
asynchronous semantics for both systems are considered.

\paragraph{Structure.}
We present the representative choreography language in which we develop our work, together with its 
associated process calculus and EPP, in Section~\ref{sec:cc-sp}.
In Section~\ref{sec:async}, we motivate and introduce the properties we would expect of an asynchronous choreography
semantics, and introduce a semantics that satisfies these properties.
We show that we can define an asynchronous variant of the target process calculus in Section~\ref{sec:asp}, and 
extend the definition of EPP towards it, preserving the precise operational correspondence from the 
synchronous case.
We relate our development to other approaches for asynchrony in choreographies in
Section~\ref{sec:related}, before concluding in Section~\ref{sec:concl} with a discussion on the implications of our
work and possible future directions.


%% file: cc-sp.tex
We review the choreography model of Minimal Choreographies (MC) and its target calculus of Stateful Processes (SP), originally introduced in~\cite{CM17:facs}.
\journal{\todo{Replace MC/Minimal Choreographies with CC/Core Choreographies. Also look at \%J\ comments.}}%

\subsection{Minimal Choreographies}

The language of MC is defined inductively in Figure~\ref{fig:cc_syntax}.
\input{figures/cc_syntax}
Processes, ranged over by $\pid p,\pid q,\pid r,\ldots$, are assumed to run concurrently.
\journal{Processes interact through two types of communication actions, which we collectively denote by $\eta$.}%
\conf{Processes interact through value communications $\gencom$, which we also denote as $\eta$.}
\journal{In value communications $\gencom$,}\conf{Here,}
process $\pid p$ evaluates expression $e$ and sends the result to $\pid q$.
The precise syntax of expressions is immaterial for our presentation; in particular, expressions can access values
stored in $\pid p$'s memory.
\journal{In label selections $\gensel$ process $\pid p$ selects, from several possible behaviours of $\pid q$, the one identified
by label $\ell$.
We abstract over the set of labels, requiring only that it be finite and contain at least two elements.}%

The remaining choreography terms denote conditionals, recursion, and termination.
In the conditional $\gencond$, process $\pid p$ evaluates expression $e$ to decide whether the choreography should
proceed as $C_1$ or as $C_2$.
In $\genrec$, we define variable $X$ to be the choreography term $C_2$, which then can be called (as $\gencall$) inside
both $C_1$ and $C_2$.
Term $\nil$ is the terminated choreography, which we sometimes omit.
For a more detailed discussion of these primitives, we refer the reader to~\cite{CM17:facs}.\footnote{We relaxed the
  syntax of MC slightly with respect to~\cite{CM17:facs} by leaving the syntax of expressions unspecified, which allows
  for the simpler conditional $\gencond$ in line with typical choreography languages. This minor change simplifies our
  presentation.}

The (synchronous) semantics of MC is a reduction semantics that uses a total state function $\sigma$ to represent the
memory state at each process $\pid p$.
Since our development is orthogonal to the details of the memory implementation, we say that $\sigma(\pid p)$ is a
representation of the memory state of $\pid p$ (left unspecified) and write $\upd\sigma{\pid p}v$ to denote the
(uniquely defined) updated memory state of $\pid p$ after receiving a value~$v$.
Term $e\eval\sigma{\pid p}v$ denotes that locally evaluating expression $e$ at $\pid p$, with memory state $\sigma$,
evaluates to $v$. We assume that expression evaluation is deterministic and always terminates. (This 
formulation captures the essence of previous memory models for choreographies, 
cf.~\cite{CM13,CM16a}.)

Transitions are defined over pairs $\langle C,\sigma\rangle$, given by the rules in
Figure~\ref{fig:cc_semantics}. As usual, we omit the angular brackets in transitions.
\input{figures/cc_semantics}
These rules are mostly standard, and we summarise their intuition.
In \rname{C}{Com}, the state of $\pid q$ is updated with the value received from $\pid p$ (which results from the
evaluation of expression $e$ at that process).
\journal{Label selections are no-ops (rule \rname{C}{Sel}), but play a role in the implementations as we discuss later.}%
Rules \rname{C}{Then} and \rname{C}{Else} are as expected, while rule \rname{C}{Ctx} allows reductions under recursive definitions.
Finally, rule \rname{C}{Struct} uses a structural precongruence $\precongr$, defined in
Figure~\ref{fig:cc_precongr}, which essentially allows (i)~independent communications to be swapped
(rule~\rname{C}{Eta-Eta}), (ii) recursive definitions to be unfolded (rule~\rname{C}{Unfold}), and (iii)~garbage
collection (rule~\rname{C}{ProcEnd}).
\input{figures/cc_precongr}
The remaining rules are additional rules allowing communications to swap with other constructs, required for achieving (i).
These rules, taken together, endow the semantics of MC with out-of-order execution: interactions not at the top level 
may be brought to the top and executed if they do not interfere with other interactions that precede them in the 
choreography's abstract syntax tree.
We write $C\equiv C'$ for $C\precongr C'$ and $C'\precongr C$, and denote the set of process names in a choreography $C$
by $\pn(C)$.

Unsurprisingly, choreographies in MC are always deadlock-free. We use this property later on, to prove that the process 
code generated from choreographies is also deadlock-free.
\begin{theorem}[Deadlock-freedom]
  \label{thm:df-by-design}
  Given a choreography $C$, either $C\precongr\nil$ (termination) or, for every $\sigma$, there exist $C'$ and $\sigma'$
  such that $C,\sigma\to C',\sigma'$.
\end{theorem}

\subsection{Stateful Processes and EndPoint Projection}

Minimal Choreographies are meant to be implemented in a minimalistic process calculus, also introduced in~\cite{CM17:facs},
called Stateful Processes (SP).
We summarize this calculus, noting that we make the same conventions and changes regarding expressions, labels, and
states as above.

The syntax of SP is reported in Figure~\ref{fig:sp_syntax}.
%
\input{figures/sp_syntax}
A term $\actor{\pid p}{\sigma_{\pid p}}{B}$ is a process with name $\pid p$, memory state $\sigma_{\pid p}$ and
behaviour $B$.
Networks, ranged over by $N,M$, are parallel compositions of processes, with $\emptyN$ being the inactive network.

Behaviours correspond to the local views of choreography actions.
The process executing a send term $\asend{\pid q}{e};B$ evaluates expression $e$ and sends the result to process
$\pid q$, proceeding as $B$.
The dual receiving behaviour $\arecv{\pid p};B$ expects a value from process $\pid p$, stores it in its memory and
proceeds as $B$.
\journal{Term $\asel{\pid q}{l};B$ sends a label $\ell$ to process $\pid q$.
Selections are received by the branching term $\abranch{\pid p}{\{ \ell_i : B_i\}_{i\in I}}$, which upon receiving one
of the labels $\ell_i$ proceeds according to $B_i$.
Branching terms must offer at least one branch.}%
The other terms are as in MC.

These intuitions are formalized in the synchronous semantics of SP, which is defined by the rules in
Figure~\ref{fig:sp_semantics}.
%
\input{figures/sp_semantics}
Rule \rname{P}{Com} models synchronous value communication: a process $\pid p$ wishing to send a value to $\pid q$ can
synchronise with a receive-from-$\pid p$ action at $\pid q$, and the state of $\pid q$ is updated accordingly.
\journal{Rule \rname{P}{Sel} is standard selection, with $\pid p$ selecting one of the branches offered by $\pid q$.}%
The remaining rules are standard.
This calculus once again includes a structural precongruence relation $\precongr$, defined in Figure~\ref{fig:sp_precongr}, which allows unfolding of recursive definition and garbage collection (removal of terminated processes).
%
\input{figures/sp_precongr}
Networks in SP do not enjoy a counterpart to Theorem~\ref{thm:df-by-design}, as send/receive actions may be unmatched or
wrongly ordered.
To avoid this happening, we focus on networks that are generated automatically from choreographies in a faithful way,
by EndPoint Projection (EPP).

EPP is defined first at the process level, by means of a partial function $\epp{C}{\pid p}$.
The rules defining behaviour projection are given in Figure~\ref{fig:cc_epp}.
\input{figures/cc_epp}
Each choreography term is projected to the local action of the process that we are projecting.
For example, a communication term $\gencom$ is projected into a send action for the sender process $\pid p$, a receive
action for the receiver process $\pid q$, or nothing otherwise.
\conf{The rule for projecting a conditional requires that the behaviours of processes not aware of the 
choice are independent of that choice (see~\cite{CM17:facs} for details).}
\journal{The rule for projecting a conditional uses the standard (and partial) merging operator $\merge$: 
$B\merge B'$ is
isomorphic to $B$ and $B'$ up to branching, where the branches of $B$ or $B'$ with distinct labels are also
included (see~\cite{CM17:facs} for details).
This aspect is found repeatedly in most choreography models~\cite{BCDLDY08,CHY12,CM13,PGGLM14}.}

\begin{definition}[EPP from MC to SP]
Given a choreography $C$ and a state $\sigma$, the endpoint projection $\epp{C,\sigma}{}$ is the parallel composition of
the EPPs of $C$ wrt all processes in $\pn(C)$:
\[
\textstyle\epp{C,\sigma}{} =
\prod_{\pid p \in \pn(C)} \actor{\pid p}{\sigma(\pid p)}{\epp{C}{\pid p}}\,.
\]
\end{definition}

Given any two states $\sigma$ and $\sigma'$, $\epp{C,\sigma}{}$ is defined iff $\epp{C,\sigma'}{}$ is defined.
If this is the case, we say that $C$ is \emph{projectable} and that $N$ is the projection of $C,\sigma$.

\begin{theorem}[EPP Theorem]
\label{thm:oc-cc-sp}
If $C$ is projectable, then, for any $\sigma$:
\begin{itemize}
\journal{
\item (Completeness) if $C,\sigma \to C',\sigma'$, then $\epp{C,\sigma}{} \to \succ \epp{C',\sigma'}{}$;
\item (Soundness) if $\epp{C,\sigma}{} \to N$, then $C,\sigma \to C',\sigma'$ for some $C'$ and $\sigma'$ such that
  $\epp{C',\sigma'}{} \prec N$.}
\conf{
\item (Completeness) if $C,\sigma \to C',\sigma'$, then $\epp{C,\sigma}{} \to \epp{C',\sigma'}{}$;
\item (Soundness) if $\epp{C,\sigma}{} \to N$, then $C,\sigma \to C',\sigma'$ for some $C'$ and $\sigma'$ such that
  $\epp{C',\sigma'}{} \precongr N$.}
\end{itemize}
\end{theorem}
\journal{
The \emph{pruning relation} $\prec$ (see~\cite{CHY12,CM13}) eliminates branches introduced by the merging operator
$\merge$ when they are not needed anymore to follow the originating choreography.
Pruning does not alter reductions, since the eliminated branches are never selected, 
cf.~\cite{CHY12,PGGLM14,LGMZ08}.}
Combining Theorem~\ref{thm:oc-cc-sp} with Theorem~\ref{thm:df-by-design} we get that the projections of choreographies
never deadlock.

\begin{corollary}[Deadlock-freedom by construction]
\label{cor:df-sp}
Let $N = \epp{C,\sigma}{}$ for some $C$ and $\sigma$. Then either $N \precongr \emptyN$ (termination) or there exists
$N'$ such that $N \to N'$.
\end{corollary}


%% file: figures/cc_syntax.tex
\begin{figure}[ht]
\noindent
\begin{align*}
  C & ::= \gencom; C \journal{\mid \gensel; C} \mid \gencond \mid  \genrec \mid \gencall \mid \nil
\end{align*}
\caption{Minimal Choreographies, syntax.}
\label{fig:cc_syntax}
\end{figure}

%% file: figures/cc_semantics.tex
\begin{figure}[ht]
\begin{eqnarray*}
&\infer[\rname{C}{Com}]
{
  \gencom;C,\sigma
  \to
  C, \upd\sigma{\pid q}v
}
{
  e\eval\sigma{\pid p}v
}
\qquad
\infer[\rname{C}{Ctx}]
{
  \genrec, \sigma \to 
  \rec{X}{C_2}{C'_1}, \sigma'
}
{
  C_1, \sigma \to C'_1, \sigma'
}
\\[1ex]
&\journal{\infer[\rname{C}{Sel}]
{
  \gensel;C, \sigma \to C, \sigma
}
{}
\quad}
\infer[\rname{C}{Then}]
{
  \gencond, \sigma \to C_1, \sigma
}
{
  e\eval\sigma{\pid p}\m{true}
}
\qquad
\infer[\rname{C}{Else}]
{
  \gencond, \sigma \to C_2, \sigma
}
{
  e\eval\sigma{\pid p}\m{false}
}
\\[1ex]
&
\infer[\rname{C}{Struct}]
{
  C_1, \sigma \to  C'_1, \sigma'
}
{
  C_1 \precongr C_2
  & C_2, \sigma \to C'_2, \sigma'
  & C'_2  \precongr C'_1
}
\end{eqnarray*}
\caption{Minimal Choreographies, synchronous semantics.}
\label{fig:cc_semantics}
\end{figure}

%% file: figures/cc_precongr.tex
\begin{figure}[ht]
\begin{eqnarray*}
&
\infer[\rname{C}{Eta-Eta}]
{
  \eta;\eta'\equiv \eta';\eta
}
{
  \pn(\eta) \cap \pn(\eta') = \emptyset
}
\qquad
\infer[\rname{C}{Eta-Rec}]
{
  \rec{X}{C_2}{(\eta;C_1)}
  \equiv 
  \eta;\genrec
}
{
  \pn(C_i) \cap \pn(\eta) = \emptyset
}
\\[1ex]
&\infer[\rname{C}{Unfold}] 
{
  \rec{X}{C_2}{C_1}
  \precongr
  \rec{X}{C_2}{C_1[C_2/X]}
}
{}
\qquad
\infer[\rname{C}{ProcEnd}] 
{
  \rec{X}{C}{\nil} \precongr \nil
}
{}
\\[1ex]
&
\infer[\rname{C}{Eta-Cond}]
{
  \cond{\pid p.e}{(\eta;C_1)}{(\eta;C_2)}
  \equiv
  \eta;\gencond
}
{
  \{ \pid p, \pid q\} \cap \pn(\eta) = \emptyset
}
\\[1ex]
&\infer[\rname{C}{Cond-Cond}]
{
  \begin{array}{c}
    \cond{\pid p.e}{
      (\cond{\pid q.e'}{C_1}{C_2})
    }{
      (\cond{\pid q.e'}{C'_1}{C'_2})
    }
    \\
    \equiv
    \\
    \cond{\pid q.e'}{
      (\cond{\pid p.e}{C_1}{C'_1})
    }{
      (\cond{\pid p.e}{C_2}{C'_2})
    }
  \end{array}
}
{
  \pid p \neq \pid q
}
\end{eqnarray*}
\caption{Minimal Choreographies, structural precongruence.}
\label{fig:cc_precongr}
\end{figure}

%% file: figures/sp_syntax.tex
\begin{figure}[ht]
\journal{
\begin{align*}
B &{}::= \asend{\pid q}{e};B \mid \arecv{\pid p};B \mid \asel{\pid q}{\ell};B \mid \abranch{\pid p}{\{ \ell_i : B_i\}_{i\in I}};B
&
N,M &{}::= \actor{\pid p}{\sigma_p}{B} \mid (N \parp M) \mid \emptyN
\\[1ex]
& \mid \cond{e}{B_1}{B_2};B \mid \genrecb \mid \gencall \mid \nil
\end{align*}}%
\conf{
\[
B ::= \asend{\pid q}{e};B \mid \arecv{\pid p};B \mid \cond{e}{B_1}{B_2};B \mid \genrecb \mid \gencall \mid \nil
\qquad
N,M ::= \actor{\pid p}{\sigma_p}{B} \mid (N \parp M) \mid \emptyN
\]
}
\caption{Stateful Processes, syntax.}
\label{fig:sp_syntax}
\end{figure}

%% file: figures/sp_semantics.tex
\begin{figure}[ht]
\begin{eqnarray*}
&\infer[\rname{P}{Com}]
  {
    \actor{\pid p}{\sigma_{\pid p}}{\asend{\pid q}{e};B_1}
    \parp
    \actor{\pid q}{\sigma_{\pid q}}{\arecv{\pid p};B_2}
    \to
    \actor{\pid p}{\sigma_{\pid p}}{B_1}
    \parp
    \actor{\pid q}{\update{\sigma_{\pid q}}v}{B_2}
  }
  {
    e\evals{\sigma_{\pid p}}v
  }
\\[1ex]
\journal{
&\infer[\rname{P}{Sel}]
  {
    \actor{\pid p}{\sigma_{\pid p}}{\asel{\pid q}{\ell_j};P}
    \parp
    \actor{\pid q}{\sigma_{\pid q}}{\abranch{\pid p}{\{ l_i : Q_i\}_{i\in I}}}
    \to
    \actor{\pid p}{\sigma_{\pid p}}{P}
    \parp
    \actor{\pid q}{\sigma_{\pid q}}{Q_j}
  }
  {j \in I}
\\[1ex]}
&
\infer[\rname{P}{Then}]
{
  \actor{\pid p}{\sigma_{\pid p}}{\cond{e}{P_1}{P_2}}
  \to
  \actor{\pid p}{\sigma_{\pid p}}{P_1}
}
{
  e\evals{\sigma_{\pid p}}\m{true}
}
\qquad
\infer[\rname{P}{Else}]
{
  \actor{\pid p}{\sigma_{\pid p}}{\cond{e}{P_1}{P_2}}
  \to
  \actor{\pid p}{\sigma_{\pid p}}{P_2}
}
{
  e\evals{\sigma_{\pid p}}\m{false}
}
\\[1ex]
&
\infer[\rname{P}{Ctx}]
  {
    \actor{\pid p}{\sigma_{\pid p}}{\rec{X}{Q}{P}} \parp N
    \to
    \actor{\pid p}{\sigma_{\pid p}}{\rec{X}{Q}{P'}} \parp N'
  }{
    \actor{\pid p}{\sigma_{\pid p}}{P} \parp N \to \actor{\pid p}{\sigma_{\pid p}}{P'} \parp N'
  }
\\[1ex]
&
\infer[\rname{P}{Struct}]
{
  N \to N'
}
{
  N \precongr M & M\to M' & M' \precongr N'
}
\qquad
\infer[\rname{P}{Par}]
  {
    N \parp M \to N' \parp M
  }
  {
    N \to  N'
  }
\end{eqnarray*}
\caption{Stateful Processes, synchronous semantics.}
\label{fig:sp_semantics}
\end{figure}

%% file: figures/sp_precongr.tex
\begin{figure}[ht]
\begin{eqnarray*}
&\infer[\rname{S}{Unfold}]
{
	\rec{X}{B_2}{B_1} \precongr \rec{X}{B_2}{B_1[B_2/X]}
}
{}
\\[1ex]
&
\infer[\rname{S}{PZero}]
{\actor{\pid p}{\sigma_{\pid p}}{\nil} \precongr \nil}
{}
\qquad
\infer[\rname{S}{NZero}]
{N \parp {\emptyN} \precongr N}
{}
\qquad
\infer[\rname{S}{ProcEnd}]
{\rec{X}{B}{\nil}\precongr\nil }
{}
\end{eqnarray*}
\caption{Stateful Processes, structural precongruence.}
\label{fig:sp_precongr}
\end{figure}

%% file: figures/cc_epp.tex
\begin{figure}[ht]
\begin{eqnarray*}
  &\epp{\gencom;C}{\pid r} =
  \begin{cases}
    \asend{\pid q}{e};\epp{C}{\pid r} & \text{if } \pid r = \pid p \\
    \arecv{\pid p};\epp{C}{\pid r} & \text{if } \pid r = \pid q \\
    \epp{C}{\pid r} & \text{otherwise}
  \end{cases}
  \qquad
  \journal{\epp{\gensel;C}{\pid r} =
  \begin{cases}
    \asel{\pid q}{\ell};\epp{C}{\pid r} & \text{if } \pid r = \pid p \\
    \abranch{\pid p}{\{ \ell : \epp{C}{\pid r} \}} & \text{if } \pid r = \pid q \\
    \epp{C}{\pid r} & \text{otherwise}
  \end{cases}}
  \\[1ex]
  &
  \epp{\gencond;C}{\pid r} =
  \begin{cases}
    \cond{e}{\epp{C_1}{\pid r}}{\epp{C_2}{\pid r}}; \epp{C}{\pid r}
    & \text{if } \pid r = \pid p \\
    \journal{( \epp{C_1}{\pid r} \merge \epp{C_2}{\pid r} );\epp{C}{\pid r} & \text{otherwise}}
    \conf{\epp{C_1}{\pid r};\epp{C}{\pid r} & \text{if } \pid r\neq\pid p \text{ and }\epp{C_1}{\pid r} = \epp{C_2}{\pid r}}
  \end{cases}
  \\[1ex]
  &\epp{\genrec}{\pid r} = 
  \rec{X}{\epp{C_2}{\pid r}}{\epp{C_1}{\pid r}}
  \qquad
  \epp{\nil}{\pid r} = \nil
  \qquad
  \epp{\gencall}{\pid r} = \gencall
\end{eqnarray*}
\caption{Minimal Choreographies, behaviour projection.}
\label{fig:cc_epp}
\end{figure}

%% file: async.tex
In order to define an asynchronous semantics for MC \journal{and SP }and prove their equivalence to the synchronous
semantics, we need to define precisely what we understand by ``asynchronous semantics'' and by ``equivalence''.
Intuitively, communications from $\pid p$ to $\pid q$ are asynchronous if: they occur in two steps (send and receive),
the send step does not require $\pid q$ to be ready, and if the send step is executed, then the corresponding receive 
step is also eventually executed.

We begin by defining asynchrony in MC.
We define the syntactic category of contexts, which can contain holes denoted as $\bullet$, formally 
defined in Figure~\ref{fig:context}.
\input{figures/context}
Structural precongruence is defined for contexts as for choreographies. We do not consider holes to be 
interactions, so they cannot be swapped with any other action.

Using contexts, we can define a function that identifies the type of the next action in context $\ctx$
involving a process $\pid r$: a communication ($\m{comm}$), a conditional ($\m{cond}$), or dependent on
how $\bullet$ is instantiated ($\bullet$).
This function is partial -- in particular, it is undefined if $\ctx$ does not contain either $\pid r$ or~$\bullet$.

\begin{definition}
  The \emph{next} action for $\pid r$ in a context $\ctx$, denoted $\hd {\ctx}{\pid r}$, is defined as follows.
  \begin{align*}
    \hd{\bullet;C}{\pid r} = \bullet &
    \qquad \hd{\nil}{\pid r} = \hd{X}{\pid r} = \mbox{undefined} \\
    \hd{\eta;\ctx}{\pid r}
    &= \begin{cases} \m{comm} & \mbox{if }\pid r\in\pn(\eta) \\ \hd{\ctx}{\pid r} & \mbox{otherwise} \end{cases}
    \\
    \hd{\cond{\pid p.e}{\ctxc[1]{\bullet}}{\ctxc[2]{\bullet}}}{\pid r}
    &= \begin{cases} \m{cond} & \mbox{if }\pid p=\pid r \\
          \hd{\ctxc[1]{\bullet}}{\pid r} & \mbox{if }\hd{\ctxc[1]{\bullet}}{\pid r}=\hd{\ctxc[2]{\bullet}}{\pid r} \\
          \mbox{undefined} & \mbox{otherwise} \end{cases}
    \\
    \hd{\rec{X}{C}{\ctx}}{\pid r}
    &= \hd{(\ctx)[C/X]}{\pid r} 
  \end{align*}
\end{definition}
The proviso in the second case of the definition of $\hd{\cond{\pid p.e}{\ctxc[1]{\bullet}}{\ctxc[2]{\bullet}}}{\pid r}$ 
ensures that the enabled action for $\pid r$ is uniquely defined.
Note that recursive definitions only need to be unfolded once.

This notion is well-defined, since structural precongruence cannot swap actions involving the same process.
We denote by $\ctx[\eta]$ the result of replacing every occurrence of $\bullet$ in $\ctx$ by $\eta$.
\journal{If $\eta$ is an interaction and $\ctx[\eta]$ is projectable, then $\hd{\ctx[\eta]}{\pid r}$ describes the top
action in the EPP of $\epp{\ctx[\eta]}{\pid r}$.}%

Since we expect asynchronous communications to occur in two steps, it is reasonable to expect MC to be extended with
additional actions.
By an \emph{extension} of MC, we understand a choreography language that has all the primitives of MC, an arbitrarily
larger set of interaction statements $\eta$, and an adequately extended function $\pn$.
Our definitions of contexts and $\m{next}$ extend automatically to these languages.

We can now formalise the intuition given at the beginning of this section.
\begin{definition}
  A semantics $\toa$ for an extension of MC is \emph{asynchronous} if the following conditions hold for any state 
$\sigma$. (We assume that $\nil;C$ stands for $C$.)
  \begin{itemize}
  \item If $\hd{\ctx}{\pid p}=\bullet$, then $\ctx[\gencom],\sigma\toa \ctx[\ast^v_{\pid q}],\sigma$,
    where $e\eval\sigma{\pid p}v$ and $\ast^v_{\pid q}$ is a statement that depends on $v$ and $\pid q$.
  \journal{\item If $\hd{\ctx}{\pid p}=\bullet$, then $\ctx[\gensel],\sigma\toa \ctx[\ast^\ell_{\pid q}],\sigma$,
    with $\ast^\ell_{\pid q}$ as above.}
  \item If $\hd{\ctx}{\pid q}=\bullet$, then $\ctx[\ast^v_{\pid q}],\sigma\toa \ctx[\nil],\upd\sigma{\pid q}v$.
  \journal{\item If $\hd{\ctx}{\pid q}=\bullet$, then $\ctx[\ast^\ell_{\pid q}],\sigma\toa \ctx[\nil],\sigma$.}
  \end{itemize}
\end{definition}

The second ingredient we need is a formal definition of equivalence. Our notion is similar to the
standard notion of operational correspondence from~\cite{G10}, but stronger, since we require that any 
additional term introduced by asynchronous reductions can always be consumed.

\begin{definition}
  An asynchronous semantics $\toa$ for an extension of MC is asynchronously equivalent to the semantics $\to$ of MC if:
  \begin{itemize}
  \item if $C,\sigma \to C',\sigma'$, then $C,\sigma \toa^{\ast} C',\sigma'$;
  \item if $C,\sigma \toa^{\ast} C',\sigma'$, then there exist $C''$ and $\sigma''$ such that
    $C,\sigma \to^{\ast} C'',\sigma''$ and $C',\sigma' \toa^{\ast} C'',\sigma''$.
  \end{itemize}
\end{definition}
Note that, in the last point of the above definition, the choreography $C'$ might include terms from the extended 
choreography language.


\journal{\todo{Change the notation for the concrete semantics arrow a}}
Now we are ready to present our extension of MC and define its asynchronous semantics.
We extend the syntax of choreographies with the runtime terms in Figure~\ref{fig:acc_syntax}. We call the resulting 
calculus aMC (for asynchronous MC).
\input{figures/acc_syntax}
The key idea is that, at runtime, a communication is expanded into multiple actions.
For example, a communication $\gencom$ expands in $\genasend$ -- a send action from $\pid p$ -- and
\journal{$\asrecv{\pid p}{\pid q}{x}$}%
\conf{$\asrecv{x}{\pid q}$} -- a receive action by $\pid q$.
\journal{The process subscripts at $\bullet$ are immaterial for the semantics, but are used later to define EPP.}%
The \journal{tag}\conf{variable} $x$ is used to specify that the original intention by the programmer was for that
message from $\pid p$ to reach that receive action at $\pid q$.
Thus, executing $\genasend$ replaces $x$ in the corresponding receive action with the actual value $v$ computed from $e$
at $\pid p$, yielding \journal{$\asrecv{\pid p}{\pid q}{v}$}\conf{$\asrecv{v}{\pid q}$}.
Finally, executing \journal{$\asrecv{\pid p}{\pid q}{v}$}\conf{$\asrecv{v}{\pid q}$} updates the state of $\pid q$.
\journal{The tags in the actions for label selections are not essential for propagating values as above, but they make the
semantics uniform.}%
Process names for the new runtime terms are defined
\journal{ as follows, ignoring process names in tags and in $\bullet$ subscripts.
\[
\pn(\genasend)=\pn(\genasels)=\{\pid p\}
\qquad
\pn(\genarecv)=\pn(\genaselr)=\{\pid q\}
\]}%
\conf{in the obvious way.}

The semantics for aMC includes:
\begin{itemize}
\item the rules from Figure~\ref{fig:acc_semantics}, replacing \rname{C}{Com}\journal{ and \rname{C}{Sel} }, and from
  Figure~\ref{fig:acc_precongr}, extending $\precongr$;
\item rules \rname{C}{Then}, \rname{C}{Else} and \rname{C}{Struct} from Figure~\ref{fig:cc_semantics};
\item the rules defining $\precongr$ (Figure~\ref{fig:cc_semantics}), with $\eta$ now ranging also over the new runtime
  terms.
\end{itemize}
\input{figures/acc_semantics}
\input{figures/acc_precongr}
By the Barendregt convention, the \journal{tags}\conf{variables} introduced by unfolding a value communication 
\journal{or selection }are globally fresh.
(We assume that their scope is global.)
This allows us to use these \journal{tags}\conf{variables} to maintain the correspondence between the value being sent 
and that being received.

The key to the asynchronous semantics of aMC lies in the new swaps allowed by $\precongr$, due to the definition of 
$\pn$ for the new terms.

\begin{example}
\label{ex:async}
Let
$C \defeq \gencom; \com{\pid p.e'}{\pid r}$.
To execute $C$ asynchronously, we must expand it:
\[
\journal{C \precongr \genasend;\ \asrecv{\pid p}{\pid q}{x};\ \assend{\pid p}{e'}{\pid r}{y};\ \asrecv{\pid p}{\pid r}{y}}
\conf{C \precongr \genasend;\ \asrecv{x}{\pid q};\ \assend{\pid p}{e'}{y};\ \asrecv{y}{\pid r}}
\]
Using $\precongr$, we can swap the second term to the end:
\[
\journal{C \precongr \genasend;\ \assend{\pid p}{e'}{\pid r}{y};\ \asrecv{\pid p}{\pid r}{y};\ \asrecv{\pid p}{\pid q}{x}}
\conf{C \precongr \genasend;\ \assend{\pid p}{e'}{y};\ \asrecv{y}{\pid r};\ \asrecv{x}{\pid q}}
\]
So $\pid p$ can send both messages immediately, and $\pid r$ can receive its message before~$\pid q$.
\end{example}

The semantics of aMC also enjoys deadlock-freedom.
\begin{theorem}[Deadlock-freedom]
  \label{thm:async-df}
  Given a choreography $C$, either $C\precongr\nil$ (termination) or, for every $\sigma$, there exist $C'$ and $\sigma'$
  such that $C,\sigma\toa C',\sigma'$.
\end{theorem}

\begin{theorem}
  The relation $\toa$ is an asynchronous semantics for aMC.
\end{theorem}
\begin{proof}[Proof (Sketch)]
  By induction on the definition of $\m{next}$.
  For the first point, observe that the only process name in $\genasend$ is $\pid p$, and therefore we can expand
  $\gencom$ and use swapping ($\precongr$) to rewrite $\ctx[\gencom]$ as $\genasend;\ctx[\asrecv{x}{\pid 
q}]$, if $\hd{\ctx}{\pid p}=\bullet$.
  The second point follows from a similar observation for \journal{$\asrecv{\pid p}{\pid q}v$}\conf{$\asrecv{v}{\pid q}$}.
\end{proof}

\begin{theorem}
  The semantics $\toa$ is asynchronously equivalent to the semantics of MC.
  \label{thm:async-equiv}
\end{theorem}
\begin{proof}[Proof (Sketch)]
Proving the first property is straightforward.
The only non-trivial case is when the transition $C,\sigma \to C',\sigma'$ uses rule $\rname{C}{Com}$, 
and in this case in the asynchronous semantics we can unfold the corresponding communication and reduce 
twice using \rname{C}{Com-S} and \rname{C}{Com-R}.

For the second property, we make two observations.
\begin{enumerate}[(1)]
\item The relation $\toa$ has the diamond property, namely: if $C,\sigma\toa C',\sigma'$ and
  $C,\sigma\toa C'',\sigma''$, then there exist $C'''$ and $\sigma'''$ such that
  $C',\sigma'\toa C''',\sigma'''$ and $C'',\sigma'' \toa C''', \sigma'''$ (see Figure~\ref{fig:proof},
  left).
  This is directly proven by case analysis on the structure of a choreography $C$ with two distinct
  possible transitions.
\item If $C,\sigma \toa C',\sigma'$, then there exist $C''$ and $\sigma''$ such that
  $C,\sigma \to^{\ast} C'',\sigma''$ and $C',\sigma' \toa^{\ast} C'',\sigma''$.
  This is similar to the proof for the first property.
  If the transition uses rule \rname{C}{Com-S}, then the choreography $C''$ is obtained by firing rule 
  \rname{C}{Com-R}, but it might be necessary to execute additional actions in case the receiver is not 
  ready for that transition yet. This is always possible due to Theorem~\ref{thm:async-df}.
\end{enumerate}

The property now follows by induction on the length of the reduction chain
$C,\sigma\toa^\ast C',\sigma'$.
When this length is $0$, the result is trivial, and when the length is $1$, the result is simply (2) above.
Otherwise, we follow the reasoning displayed in Figure~\ref{fig:proof}, right. 
By assumption, we have that $C,\sigma \toa^{\ast} C_1,\sigma_1 \toa C',\sigma'$ (top of the picture).
By induction hypothesis, there exist $C'_1$ and $\sigma'_1$ such that $C,\sigma \to^{\ast} C'_1,\sigma'_1$ and 
$C_1,\sigma_1 \toa^{\ast} C'_1,\sigma'_1$ (left triangle).
By iterated application of (1) above, there exist $C''_1$ and $\sigma''_1$ such that
$C'_1,\sigma'_1\toa^{?} C''_1,\sigma''_1$ and $C',\sigma' \toa^{\ast} C''_1,\sigma''_1$ (right square).
The reduction $C'_1,\sigma'_1 \toa^{?} C''_1,\sigma''_1$ is optional because it may be the case that the
action in the reduction $C_1,\sigma_1 \toa C',\sigma'$ is already included in the reduction chain
$C_1,\sigma_1 \toa^{\ast} C'_1,\sigma'_1$.
In this case, the thesis follows, taking $C''=C'_1$ and $\sigma''=\sigma'_1$.
Otherwise, we apply (2) above (lower right triangle) to obtain $C''$ and $\sigma''$ such that
$C'_1,\sigma'_1 \to^{\ast} C'',\sigma''$ and $C''_1,\sigma''_1 \toa^{\ast} C'',\sigma''$, and the thesis
follows.
\end{proof}

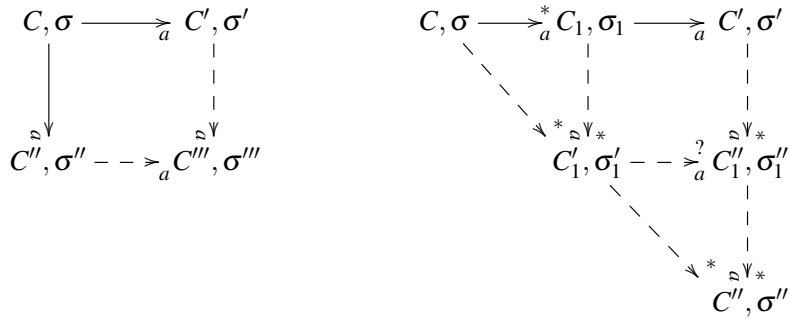
\begin{figure}[ht]
\[
\xymatrix@R+1em{%
  C,\sigma \ar[r]_(.7)a \ar[d]_(.8)[right]a
  & \ C',\sigma' \ar@{-->}[d]_(.8)[right]a
  &&
  C,\sigma \ar[r]^(.7)\ast_(.7)a \ar@{-->}[dr]+UL^(1)\ast
  & \ C_1,\sigma_1 \ar[r]_(.7)a \ar@{-->}[d]^(.8)[right]\ast_(.8)[right]a
  & \ C',\sigma' \ar@{-->}[d]^(.8)[right]\ast_(.8)[right]a
  \\
  \rule[2ex]{0mm}{0mm} C'',\sigma'' \ar@{-->}[r]_(.7)a
  & \rule[2ex]{0mm}{0mm}\ C''',\sigma'''
  &&
  & \rule[2ex]{0mm}{0mm} C'_1,\sigma'_1  \ar@{-->}[r]^(.7)?_(.7)a \ar@{-->}[dr]+UL^(1)\ast
  & \rule[2ex]{0mm}{0mm}\ C''_1,\sigma''_1 \ar@{-->}[d]^(.8)[right]\ast_(.8)[right]a
  \\
  &
  &&
  && \rule[2ex]{0mm}{0mm}\ C'',\sigma''
}
\]

\caption{The diamond property for $\toa$ (left) and the reasoning in the proof of the second property in
  Theorem~\ref{thm:async-equiv} (right).}
\label{fig:proof}
\end{figure}

\begin{remark}
Reasoning on asynchronous behaviour is known to be hard for programmers.
That is why our terms for modelling asynchronous communications are \emph{runtime terms}: they are not intended to be 
part of the source language used by developers to program.
In general, programmers do not need to be aware of the development in this section, thanks to
Theorem~\ref{thm:async-equiv}.
Instead, they can just write a choreography thinking in terms of synchronous communications as usual (syntax and
semantics), and then just assume that adopting an asynchronous communications will not lead to any bad behaviour.
In the next section, we show that this abstraction carries over to asynchronous process code generated by EPP.
Since EPP is automatic, developers do not need to worry about asynchrony in process code either.
\end{remark}


%% file: figures/context.tex
\begin{figure}[ht]
\noindent
\begin{align*}
  \ctx & ::= \bullet; C \mid \eta; \ctx \mid
  \cond{\pid p.e}{\ctxc[1]{\bullet}}{\ctxc[2]{\bullet}} \mid \rec{X}{C}{\ctx} \mid \gencall \mid \nil
\end{align*}
\caption{Contexts for asynchrony.}
\label{fig:context}
\end{figure}

%% file: figures/acc_syntax.tex
\begin{figure}[ht]
\[
\eta ::= \ldots
\mid \genasend
\mid \genarecv
\journal{\mid \genasels
\mid \genaselr}
\qquad
\maybe v ::= x \mid v
\journal{\qquad
\maybe \ell ::= x \mid \ell}
\]
\caption{Asynchronous Minimal Choreographies, runtime terms.}
\label{fig:acc_syntax}
\end{figure}

%% file: figures/acc_semantics.tex
\begin{figure}[ht]
\begin{eqnarray*}
&\infer[\rname{C}{Com-S}]
{
  \genasend;C,\sigma
  \toa
  C[v/x], \sigma
}
{
  e\eval\sigma{\pid p}v
}
\qquad\infer[\rname{C}{Com-R}]
{
  \journal{\asrecv{\pid p}{\pid q}v}\conf{\asrecv v{\pid q}};C,\sigma
  \toa
  C, \upd\sigma{\pid q}v
}
{
}
\journal{\\[1ex]
&\infer[\rname{C}{Sel-S}]
{
  \genasels;C,\sigma \toa G,C[\ell/x],\sigma
}
{
}
\qquad
\infer[\rname{C}{Sel-R}]
{
  \acom{\bullet_{\pid p}}{\pid q[\ell]}\ell;C,\sigma \toa C,\sigma
}
{
}}
\end{eqnarray*}
\caption{Asynchronous Minimal Choreographies, semantics (runtime terms).}
\label{fig:acc_semantics}
\end{figure}

%% file: figures/acc_precongr.tex
\begin{figure}[ht]
\[
\infer[\rname{C}{Com-Unfold}]
{
  \gencom \precongr \genasend;\journal{\asrecv{\pid p}{\pid q}x}\conf{\asrecv x{\pid q}}
}
{\rule{0cm}{1em}}
\journal{\qquad
\infer[\rname{C}{Sel-Unfold}]
{
  \gensel \precongr \genasels;\acom{\bullet_{\pid p}}{\pid q[\ell]}x
}
{}}
\]
\caption{Asynchronous Minimal Choreographies, new rule for structural precongruence.}
\label{fig:acc_precongr}
\end{figure}

%% file: asp.tex
In the previous section, we characterized asynchrony at an abstract level, and showed that our
choreography model satisfies the properties we identified.
In this section, we take a more down-to-earth approach: we define directly an asynchronous variant
for SP, extend EPP to the runtime terms introduced earlier, and show that the asynchronous semantics
for both calculi again satisfy the EPP Theorem (Theorem~\ref{thm:oc-cc-sp}).
\journal{\todo{check syntax of runtime terms}}

\subsection{Asynchronous Stateful Processes}

Syntactically, we need to make a slight change to our processes: each process is now also equipped
with a queue of incoming messages.
We name this version of the calculus aSP.
So actors in the asynchronous model have the form $\actor[\rho]{\pid p}\sigma B$, where $\rho$ is a queue
of incoming messages. We sometimes omit $\rho$ when it is empty; in particular, this allows us to view SP networks as 
special cases of aSP networks. A message is a pair $\langle\pid q,m\rangle$, where $\pid q$ is the sender process and 
$m$ is a value, label, or process identifier.

The semantics of aSP consists of rules \rname{P}{Then}, \rname{P}{Else}, \rname{P}{Par}, \rname{P}{Ctx} and
\rname{P}{Struct} from SP (Figure~\ref{fig:sp_semantics}) together with the new rules given in
Figure~\ref{fig:asp_semantics}.
Structural precongruence for aSP is defined exactly as for SP.
%
\input{figures/asp_semantics}
We write $\rho \cdot \langle\pid q,m\rangle$ to denote the queue obtained by appending message
$\langle\pid q,m\rangle$ to $\rho$, and $\langle\pid q,m\rangle \cdot \rho$ for the queue with
$\langle\pid q,m\rangle$ at the head and $\rho$ as tail.
We simulate having one separate FIFO queue for each other process by allowing incoming messages from different
senders to be exchanged, which we represent using the congruence $\rho\preceq\rho'$ defined by the rule
$\langle\pid p,m\rangle \cdot \langle\pid q,m'\rangle \preceq \langle\pid q,m'\rangle \cdot \langle\pid p,m\rangle$ if 
$\pid p \neq \pid q$.

All behaviours of SP are valid also in aSP.
\begin{theorem}
  \label{thm:pp-vs-app}
  Let $N$ be an SP network.
  If $N\to N'$ (in SP), then $N_{[]}\to^*N'_{[]}$ (in aSP), where $N_{[]}$ denotes the asynchronous
  network obtained by adding an empty queue to each process.
\end{theorem}
\begin{proof}
  Straightforward by case analysis: the only reduction rule in SP that is not present in aSP is that for
  communications, which can be simulated by applying rules \rname{P}{Com-S} and \rname{P}{Com-R} in
  sequence.
\end{proof}

The converse is not true, so the relation between SP and aSP is not so strong as that between MC and aMC (stated in 
Theorem~\ref{thm:async-equiv}).
This is because of deadlocks: in SP, a communication action can only take place when the sender and receiver are
ready to synchronize; in aSP, a process can send a message to another process, even though the
intended recipient is not yet able to receive it.
For example, the network $\actor{\pid p}{\sigma_{\pid p}}{\asend{\pid q}1}\parp\actor{\pid q}{\sigma_{\pid q}}{\nil}$
is deadlocked in SP, but its counterpart in aSP (obtained by adding empty queues at each process) reduces to 
$\actor[\langle \pid p, 1 \rangle]{\pid q}{\sigma_{\pid q}}{\nil}$. (This network is not equivalent to $\nil$, since 
there is a non-empty queue.)

More interestingly, the network
\begin{equation}
   \label{eqn:deadlock}
   \actor{\pid p}{\sigma_{\pid p}}{\asend{\pid q}e_1;\arecv{\pid r}} \parp
   \actor{\pid q}{\sigma_{\pid q}}{\asend{\pid r}e_2;\arecv{\pid p}} \parp
   \actor{\pid r}{\sigma_{\pid r}}{\asend{\pid p}e_3;\arecv{\pid q}}
 \end{equation}
is deadlocked in SP, but reduces to $\nil$ in aSP (all queues are eventually emptied).
It is possible to extend MC with communication primitives that capture this type of behaviour, as
discussed in~\cite{CLM17}; we briefly discuss this point in Section~\ref{sec:concl}.

\subsection{Asynchronous EndPoint Projection}

Defining an EPP from aMC to aSP requires extending the previous definition with clauses for the new
runtime terms, which populate the local queues in the projections.
Intuitively, when compiling, e.g.,~$\asrecv{v}{\pid q}$, we add a message containing $v$ at the top of 
$\pid q$'s queue.

There are two problems with this approach.
The first is a syntactic issue: each message in the queues of aSP processes must include the name of its sender, but 
that information is not present in the runtime term $\asrecv v{\pid q}$.
This is a minor issue that can be dealt with by annotating send and receive actions with the name of the process
sending them -- writing, e.g., $\asrecv{\langle\pid p,x\rangle}{\pid q}$ and~$\asrecv{\langle\pid p,v\rangle}{\pid q}$.
Changing the syntax of runtime terms and rule~\rname{C}{Com-Unfold} in this way is trivial, and we assume this
annotated terms in the remainder of this session.\footnote{We could have defined the runtime terms for aMC annotated
  from the start. However, we felt that it would be unnatural to include them in a choreographic presentation, as they
  are only needed for projecting runtime terms -- a feature that is a technicality required for the proof of
  Theorem~\ref{thm:aepp} below, since the programmer should not write runtime terms anyway.}

The second problem arises because we can write choreographies that use runtime terms in a ``wrong'' way, for which
Theorem~\ref{thm:oc-cc-sp} no longer holds.

\begin{example}
  \label{ex:problem}
  Consider the choreography $C=\com{\pid p.1}{\pid q};\asrecv{\langle\pid p,2\rangle}{\pid q}$.
  If we naively project it as described informally, we obtain
  $
  \actor[{[]}]{\pid p}{\sigma(\pid p)}{\asend{\pid q}1}\parp\actor[\langle\pid p,2\rangle]{\pid q}{\sigma(\pid q)}{\arecv{\pid p};\arecv{\pid p}}\,,
  $
where $[]$ is the empty queue, and $\pid q$ will receive $2$ before it receives $1$.
\end{example}

To avoid this undesired behaviour, we restrict ourselves to \emph{well-formed} choreographies: those that can arise
from executing a choreography that does not contain runtime terms (i.e., a program). Since runtime terms are supposed 
to be hidden from the programmer anyway, this restriction does not make us lose any generality in practice.

\begin{definition}[Well-formedness]
  A choreography $C$ in aMC containing runtime terms is \emph{well-formed} if
  $\eta_1;\ldots;\eta_n;C^{MC} \precongr^- C$, where:
  \begin{itemize}
  \item $\precongr^-$ is structural precongruence without rule \rname{C}{Unfold};
  \item each $\eta_i$ is an instantiated receive action of the form $\asrecv{\langle\pid p_i,v_i\rangle}{\pid q_i}$%
    \journal{ or $\acom{\bullet_{\pid p}}{\pid q[\ell]}\ell$};
  \item $C^{MC}$ is an MC choreography (i.e., a choreography without runtime terms).
  \end{itemize}
\end{definition}
Well-formedness is decidable, since the set of choreographies equivalent up to $\precongr^-$ is decidable.
More efficiently, one can check that $C$ is well-formed by swapping all runtime actions to the beginning and
folding all paired send/receive terms.
Furthermore, choreography execution preserves well-formedness; in particular, the problematic choreography
from Example~\ref{ex:problem} is not well-formed. More generally, we can 
use well-formedness in the remainder of this section to reason about EPP.

\begin{definition}[Asynchronous EPP from aMC to aSP]
  Let $C$ be a well-formed aMC choreography and $\sigma$ be a state.
  Without loss of generality, we assume that $C$ does not contain $\genasend$ actions.\footnote{By well-formedness, we
  can always rewrite $C$ to an equivalent choreography satisfying this condition; such a choreography is also
  guaranteed not to contain $\asrecv x{\pid q}$ actions.}
  The EPP of $C$ and $\sigma$ is defined as
  \[
  \textstyle\epp{C,\sigma}{} =
  \prod_{\pid p \in \pn(C)} \actor[\stepp{C}{\pid p}]{\pid p}{\sigma(\pid p)}{\epp{C}{\pid p}}
  \]
  where $\epp{C}{\pid p}$ is defined as in Figure~\ref{fig:cc_epp} with the extra rule in
  Figure~\ref{fig:aepp} and $\stepp{C}{\pid p}$ is defined by the rules in
  Figure~\ref{fig:aepp_state}.
\end{definition}
\input{figures/aepp}
\input{figures/aepp_state}
In the last case of the definition of $\stepp{C}{\pid p}$ (bottom-right), $\eta$ ranges over all
cases that are not covered previously.
The rule for the conditional may seem a bit surprising:
\journal{in the case of projectable choreographies,
mergeability and well-formedness together imply that unmatched receive actions 
at a process must occur in the same order in both branches.}%
\conf{projectability of choreographies implies that unmatched receive actions at a process must occur in 
the same order in both branches.}
We could alternatively define projection only for well-formed choreographies in the ``canonical
form'' implicit in the definition of well-formedness.

With this definition, we can state an asynchronous variant of Theorem~\ref{thm:oc-cc-sp}.
\begin{theorem}[Asynchronous EPP Theorem]
  \label{thm:aepp}
  If $C$ is a projectable and well-formed MC choreography, then, for all $\sigma$:
  \begin{itemize}
  \item (Completeness) if $C,\sigma \to C',\sigma'$, then $\epp{C,\sigma}{} \to \epp{C',\sigma'}{}$;
  \item (Soundness) if $\epp{C,\sigma}{} \to N$, then $C,\sigma \to C',\sigma'$ for some $C'$ and $\sigma'$
    such that $\epp{C',\sigma'}{} \precongr N$.
  \end{itemize}
\end{theorem}
As a consequence, Corollary~\ref{cor:df-sp} applies also to the asynchronous case: the processes
projected from MC into aSP are deadlock-free, even when they contain runtime terms.


As usual, the hypotheses in Theorem~\ref{thm:aepp} are not necessary, and this result also holds for some choreographies that are not well-formed.
For example, the network in~\eqref{eqn:deadlock}
\conf{can be written as the projection of a choreography with runtime terms}%
\journal{
is the projection of
\[
\assend{\pid p}{e_1}{\pid q};
\assend{\pid q}{e_2}{\pid r};
\assend{\pid r}{e_3}{\pid p};
\asrecv{\pid p}{\pid r};
\asrecv{\pid q}{\pid p};
\asrecv{\pid r}{\pid q};
\]
with an adequately defined state $\sigma$}, but this choreography is not structurally precongruent to
any choreography obtained from execution of an MC choreography.


%% file: figures/asp_semantics.tex
\begin{figure}[ht]
\begin{eqnarray*}
&\infer[\rname{P}{Com-S}]
{
  \actor[\rhop]{\pid p}{\sigma_{\pid p}}{\asend{\pid q}{e};B_{\pid p}} \parp
  \actor[\rhoq]{\pid q}{\sigma_{\pid q}}{B_{\pid q}}
  \to
  \actor[\rhop]{\pid p}{\sigma_{\pid p}}{B_{\pid p}} \parp \actor[\rho'_{\pid q}]{\pid q}{\sigma_{\pid q}}{B_{\pid q}}
}
{
  e\evals{\sigma_{\pid p}}v
  &
  \rho'_{\pid q}=\rhoq\cdot\langle\pid p,v\rangle
}
\qquad
\infer[\rname{P}{Com-R}]
{
  \actor[\rhoq]{\pid q}{\sigma_{\pid q}}{\arecv{\pid p};B}
  \to
  \actor[\rho'_{\pid q}]{\pid q}{\update{\sigma_{\pid q}}{v}}{B}
}
{
  \rhoq\preceq\langle\pid p,v\rangle\cdot\rho'_{\pid q}
}
\end{eqnarray*}
\caption{Asynchronous Stateful Processes, semantics (new rules).}
\label{fig:asp_semantics}
\end{figure}

%% file: figures/aepp.tex
\begin{figure}[ht]
\[
\epp{\asrecv{\langle\pid p,v\rangle}{\pid q};C}{\pid r}=
  \begin{cases}
    \arecv{\pid p};\epp{C}{\pid r} & \text{if } \pid r = \pid q \\
    \epp{C}{\pid r} & \text{otherwise}
  \end{cases}
\]
\caption{Minimal Choreographies, asynchronous behaviour projection (new rule).}
\label{fig:aepp}
\end{figure}

%% file: figures/aepp_state.tex
\begin{figure}[ht]
\begin{eqnarray*}
&\stepp{\asrecv{\langle\pid p,v\rangle}{\pid q};C}{\pid r}=
  \begin{cases}
    \langle\pid p,v\rangle\cdot\stepp{C}{\pid r} & \text{if } \pid r = \pid q \\
    \stepp{C}{\pid r} & \text{otherwise}
  \end{cases}
\qquad
\begin{array}l
\stepp{\gencond;C}{\pid r} = \stepp{C_1}{\pid r}\cdot\stepp{C}{\pid r}\\[1ex]
\stepp{\eta;C}{\pid r} = \stepp{C}{\pid r}
\end{array}
\end{eqnarray*}
\caption{Minimal Choreographies, projection of messages in transit.}
\label{fig:aepp_state}
\end{figure}

%% file: related.tex
The first work introducing an asynchronous semantics to choreographies is~\cite{CM13}, as described in the
introduction.
The same approach was later adapted to compositional choreographies~\cite{MY13} and multiparty session
types~\cite{HYC16}.
The models presented in these works all suffer from the shortcoming discussed in the introduction: choreographies can
reduce to states that would normally not be reachable in the real world.

We can now give a formal example of the consequences of this discrepancy.
\begin{example}
Let $C$ be the choreography from the introduction: $C \defeq \com{\pid p.1}{\pid q};\com{\pid p.2}{\pid r}$.
If we adopted the asynchronous semantics from~\cite{CM13}, then the reduction
\[C,\sigma \to \com{\pid p.1}{\pid q}, \upd\sigma{\pid r}2\]
would be possible for any $\sigma$.
This shows that Theorem~\ref{thm:aepp} (the completeness direction) does not hold in this semantics, since
$\epp{C,\sigma}{}$ cannot reduce to any $N$ such that $\epp{\com{\pid p.1}{\pid q}, \upd\sigma{\pid r}2}{} \precongr N$
(we have to consume the first output by $\pid p$ first).

Therefore, if we want to formalise the correspondence between a choreography and its EPP in this setting, we need to 
consider multiple steps.
In this example, specifically, we can observe that
\[\com{\pid p.1}{\pid q}, \upd\sigma{\pid r}2 \to \nil, \upd{\upd\sigma{\pid r}2}{\pid q}1\]
and that $\epp{C,\sigma}{} \to^* \epp{\nil,\upd{\upd\sigma{\pid r}2}{\pid q}1}{}$.
\end{example}

In general, the EPP Theorem with this kind of asynchronous semantics is weaker and more complicated.
We report it here by adapting the formulation from~\cite[Chapter 2]{M13:phd} (the full version of~\cite{CM13}) to our
notation:
\begin{theorem}
  \label{thm:aepp-bad}
  If $C$ is a projectable choreography, then, for all $\sigma$:
  \begin{itemize}
  \item (Completeness) if $C,\sigma \to C',\sigma'$, then $C',\sigma' \to^* C'',\sigma''$ for some $C'',\sigma''$
and $\epp{C,\sigma}{} \to^* \epp{C'',\sigma''}{}$;
  \item (Soundness) if $\epp{C,\sigma}{} \to^* N$, then $N \to^* N'$ for some $N'$, and $C,\sigma \to^* 
C',\sigma'$ for some $\sigma'$ and $C'$
    such that $\epp{C',\sigma'}{} \precongr N'$.
  \end{itemize}
\end{theorem}
Compared to our Theorem~\ref{thm:aepp} above, Theorem~\ref{thm:aepp-bad} is missing the step-by-step correspondence
between choreographies and their projections, and therefore does not say much about the intermediate steps.

In~\cite{DY13}, multiparty session types are equipped with runtime terms that represent messages in transit
in asynchronous communications, similarly to our approach. Differently from our model, the semantics 
in~\cite{DY13} uses a labelled transition system (instead of out-of-order execution) to identify when a
communication can be executed asynchronously. The difference between these two systems seems to be mostly
a matter of presentation, but their expressive powers are difficult to compare formally because a counterpart 
to Theorem~\ref{thm:async-equiv} is not provided for the model of~\cite{DY13}.

A more recent development that is nearer to ours is the model of Applied Choreographies~\cite{GGM15}, where
out-of-order execution (modelled as structural equivalences) is used to swap independent partial choreographic
actions, which contain terms that are similar to our runtime asynchronous terms.
However, the asynchronous terms in Applied Choreographies are not enough to ensure that the semantics is sound;
therefore, these choreographies also need to include queues to store messages in transit.
This makes their development substantially more complicated than the one we propose.
Furthermore, the work in~\cite{GGM15} focuses on implementation models, and does not provide a formal definition of
asynchrony for choreographies as in this paper.
There is also no correspondence result relating Applied Choreographies to a standard synchronous choreography
semantics, although we conjecture that it is possible to map a fragment of that language into an asynchronous
extension of MC in the sense of our definition.

Our approach relies on out-of-order execution for choreographic interactions, which was first introduced in~\cite{CM13} 
to capture parallel execution. We extended it to support the swapping of interactions supported by 
asynchronous message passing. Out-of-order execution does not introduce any burden on the 
programmer, since only non-interfering interactions can be swapped and all swappings are thus safe by design (more 
precisely, swaps correspond to the parallel semantics of typical process calculi~\cite{M13:phd}); instead, they just 
model different safe ways of executing the same choreography.


%% file: conclusions.tex
We presented a definition of asynchrony in the minimal choreography language MC and showed that we can define an
asynchronous semantics for an extension of MC (aMC) that respects this definition.
This construction was deliberately made in a modular way (it changes only the rules for communications, leaving the 
rest untouched), so that it can be easily applied to more expressive languages.
We discuss a few relevant cases.

The first case regards label selection, usually written $\gensel$, which is a widespread primitive in 
choreographies~\cite{CHY12}. In a selection, the sender informs the receiver of a local choice, and the receiver can 
use this information to change its behaviour. Selections are mostly relevant for extending the domain of EPP, 
which is orthogonal to this work. Our approach applies directly also to selections, by adding asynchronous runtime terms 
that are similar to the ones for value communications but carry selection labels instead of values. Name 
mobility~\cite{CM13,CM17:forte} can also be treated in a similar way.

Instead of out-of-order execution, some choreography models also include explicit parallel composition, e.g., $C\parp 
C'$~\cite{CHY12,QZCY07}. Most behaviours of $C \parp C'$ are captured by out-of-order execution, for example
$\com{\pid p.e}{\pid q} \parp \com{\pid r.e'}{\pid s}$ is equivalent to
$\com{\pid p.e}{\pid q}; \com{\pid r.e'}{\pid s}$ in MC (due to rule~\rname{C}{Eta-Eta} in
Figure~\ref{fig:cc_semantics}) -- see~\cite{CM13} for a deeper discussion.
Generalising our construction to supporting also an explicit parallel operator is straightforward (structural 
precongruence is extended homomorphically without surprises).

A more interesting extension is multicom~\cite{CLM17}, a primitive that considerably extends the 
expressivity of choreographies by allowing for general criss-cross communications. A prime example is the asynchronous 
exchange $\{\gencom, \com{\pid q.e'}{\pid p}\}$, which allows two processes to exchange message without waiting for each 
other (this is the building block, e.g., of the alternating 2-bit protocol given in~\cite{CLM17}).
Multicoms crucially depend on asynchrony, making them an interesting case study in our context.
Our development applies immediately to multicoms, since we can just apply the expansion rule $\rname{C}{Com-Unfold}$ to 
all communications in a multicom simultaneously.

Likewise, we conjecture that our development is applicable to the models in~\cite{DGGLM15,LNN16}.

%% file: main.bbl
\begin{thebibliography}{10}
\providecommand{\bibitemdeclare}[2]{}
\providecommand{\surnamestart}{}
\providecommand{\surnameend}{}
\providecommand{\urlprefix}{Available at }
\providecommand{\url}[1]{\texttt{#1}}
\providecommand{\href}[2]{\texttt{#2}}
\providecommand{\urlalt}[2]{\href{#1}{#2}}
\providecommand{\doi}[1]{doi:\urlalt{http://dx.doi.org/#1}{#1}}
\providecommand{\bibinfo}[2]{#2}

\bibitemdeclare{proceedings}{forte2016}
\bibitem{forte2016}
\bibinfo{editor}{Elvira \surnamestart Albert\surnameend} \&
  \bibinfo{editor}{Ivan \surnamestart Lanese\surnameend}, editors
  (\bibinfo{year}{2016}): \emph{\bibinfo{title}{Formal Techniques for
  Distributed Objects, Components, and Systems -- 36th IFIP WG 6.1
  International Conference, FORTE 2016}}. {\sl \bibinfo{series}{LNCS}}
  \bibinfo{volume}{9688}, \bibinfo{publisher}{Springer},
  \doi{10.1007/978-3-319-39570-8}.

\bibitemdeclare{misc}{BPMN}
\bibitem{BPMN}
\emph{\bibinfo{title}{{B}usiness {P}rocess {M}odel and {N}otation}}.
\newblock \bibinfo{howpublished}{\url{http://www.omg.org/spec/BPMN/2.0/}}.

\bibitemdeclare{article}{CHY12}
\bibitem{CHY12}
\bibinfo{author}{Marco \surnamestart Carbone\surnameend},
  \bibinfo{author}{Kohei \surnamestart Honda\surnameend} \&
  \bibinfo{author}{Nobuko \surnamestart Yoshida\surnameend}
  (\bibinfo{year}{2012}): \emph{\bibinfo{title}{Structured
  Communication-Centered Programming for Web Services}}.
\newblock {\sl \bibinfo{journal}{ACM Trans.\ Program.\ Lang.\ Syst.}}
  \bibinfo{volume}{34}(\bibinfo{number}{2}), pp. \bibinfo{pages}{8:1--8:78},
  \doi{10.1145/2220365.2220367}.

\bibitemdeclare{inproceedings}{CM13}
\bibitem{CM13}
\bibinfo{author}{Marco \surnamestart Carbone\surnameend} \&
  \bibinfo{author}{Fabrizio \surnamestart Montesi\surnameend}
  (\bibinfo{year}{2013}): \emph{\bibinfo{title}{Deadlock-freedom-by-design:
  multiparty asynchronous global programming}}.
\newblock In \bibinfo{editor}{Roberto \surnamestart Giacobazzi\surnameend} \&
  \bibinfo{editor}{Radhia \surnamestart Cousot\surnameend}, editors: {\sl
  \bibinfo{booktitle}{POPL}}, \bibinfo{publisher}{{ACM}}, pp.
  \bibinfo{pages}{263--274}.
\newblock \urlprefix\url{http://doi.acm.org/10.1145/2429069.2429101}.

\bibitemdeclare{misc}{chor:website}
\bibitem{chor:website}
\bibinfo{author}{\surnamestart Chor\surnameend}:
  \emph{\bibinfo{title}{{Programming Language}}}.
\newblock \bibinfo{note}{\url{http://www.chor-lang.org/}}.

\bibitemdeclare{inproceedings}{CLM17}
\bibitem{CLM17}
\bibinfo{author}{Lu{\'\i}s \surnamestart Cruz{-}Filipe\surnameend},
  \bibinfo{author}{Kim~S. \surnamestart Larsen\surnameend} \&
  \bibinfo{author}{Fabrizio \surnamestart Montesi\surnameend}
  (\bibinfo{year}{2017}): \emph{\bibinfo{title}{The Paths to Choreography
  Extraction}}.
\newblock In \bibinfo{editor}{Javier \surnamestart Esparza\surnameend} \&
  \bibinfo{editor}{Andrzej~S. \surnamestart Murawski\surnameend}, editors: {\sl
  \bibinfo{booktitle}{FoSSaCS}}, {\sl \bibinfo{series}{LNCS}}
  \bibinfo{volume}{10203}, \bibinfo{publisher}{Springer}, pp.
  \bibinfo{pages}{424--440}, \doi{10.1007/978-3-662-54458-7\_25}.

\bibitemdeclare{inproceedings}{CM16a}
\bibitem{CM16a}
\bibinfo{author}{Lu{\'\i}s \surnamestart Cruz-Filipe\surnameend} \&
  \bibinfo{author}{Fabrizio \surnamestart Montesi\surnameend}
  (\bibinfo{year}{2016}): \emph{\bibinfo{title}{Choreographies in Practice}}.
\newblock In \bibinfo{editor}{Albert} \& \bibinfo{editor}{Lanese}
  \cite{forte2016}, pp. \bibinfo{pages}{114--123},
  \doi{10.1007/978-3-319-39570-8\_8}.

\bibitemdeclare{inproceedings}{CM17:facs}
\bibitem{CM17:facs}
\bibinfo{author}{Lu{\'\i}s \surnamestart Cruz-Filipe\surnameend} \&
  \bibinfo{author}{Fabrizio \surnamestart Montesi\surnameend}
  (\bibinfo{year}{2017}): \emph{\bibinfo{title}{A Core Model for Choreographic
  Programming}}.
\newblock In \bibinfo{editor}{Olga \surnamestart Kouchnarenko\surnameend} \&
  \bibinfo{editor}{Ramtin \surnamestart Khosravi\surnameend}, editors: {\sl
  \bibinfo{booktitle}{FACS}}, {\sl \bibinfo{series}{LNCS}}
  \bibinfo{volume}{10231}, \bibinfo{publisher}{Springer}, pp.
  \bibinfo{pages}{17--35}, \doi{10.1007/978-3-319-57666-4\_3}.

\bibitemdeclare{inproceedings}{CM17:forte}
\bibitem{CM17:forte}
\bibinfo{author}{Lu{\'\i}s \surnamestart Cruz-Filipe\surnameend} \&
  \bibinfo{author}{Fabrizio \surnamestart Montesi\surnameend}
  (\bibinfo{year}{2017}): \emph{\bibinfo{title}{Procedural Choreographic
  Programming}}.
\newblock In \bibinfo{editor}{Ahmed \surnamestart Bouajjani\surnameend} \&
  \bibinfo{editor}{Alexandra \surnamestart Silva\surnameend}, editors: {\sl
  \bibinfo{booktitle}{FORTE 2017}}, {\sl \bibinfo{series}{LNCS}}
  \bibinfo{volume}{10321}, \bibinfo{publisher}{Springer}, pp.
  \bibinfo{pages}{92--107}, \doi{10.1007/978-3-319-60225-7\_7}.

\bibitemdeclare{inproceedings}{DGGLM15}
\bibitem{DGGLM15}
\bibinfo{author}{Mila \surnamestart Dalla~Preda\surnameend},
  \bibinfo{author}{Maurizio \surnamestart Gabbrielli\surnameend},
  \bibinfo{author}{Saverio \surnamestart Giallorenzo\surnameend},
  \bibinfo{author}{Ivan \surnamestart Lanese\surnameend} \&
  \bibinfo{author}{Jacopo \surnamestart Mauro\surnameend}
  (\bibinfo{year}{2015}): \emph{\bibinfo{title}{Dynamic Choreographies -- Safe
  Runtime Updates of Distributed Applications}}.
\newblock In \bibinfo{editor}{Tom \surnamestart Holvoet\surnameend} \&
  \bibinfo{editor}{Mirko \surnamestart Viroli\surnameend}, editors: {\sl
  \bibinfo{booktitle}{COORDINATION}}, {\sl \bibinfo{series}{LNCS}}
  \bibinfo{volume}{9037}, \bibinfo{publisher}{Springer}, pp.
  \bibinfo{pages}{67--82}, \doi{10.1007/978-3-319-19282-6\_5}.

\bibitemdeclare{inproceedings}{DY13}
\bibitem{DY13}
\bibinfo{author}{Pierre-Malo \surnamestart Deni{\'e}lou\surnameend} \&
  \bibinfo{author}{Nobuko \surnamestart Yoshida\surnameend}
  (\bibinfo{year}{2013}): \emph{\bibinfo{title}{Multiparty Compatibility in
  Communicating Automata: Characterisation and Synthesis of Global Session
  Types}}.
\newblock In \bibinfo{editor}{Fedor~V. \surnamestart Fomin\surnameend},
  \bibinfo{editor}{Rusins \surnamestart Freivalds\surnameend},
  \bibinfo{editor}{Marta~Z. \surnamestart Kwiatkowska\surnameend} \&
  \bibinfo{editor}{David \surnamestart Peleg\surnameend}, editors: {\sl
  \bibinfo{booktitle}{ICALP (2)}}, {\sl \bibinfo{series}{LNCS}}
  \bibinfo{volume}{7966}, \bibinfo{publisher}{Springer}, pp.
  \bibinfo{pages}{174--186}, \doi{10.1007/978-3-642-39212-2\_18}.

\bibitemdeclare{article}{GGM15}
\bibitem{GGM15}
\bibinfo{author}{Maurizio \surnamestart Gabbrielli\surnameend},
  \bibinfo{author}{Saverio \surnamestart Giallorenzo\surnameend} \&
  \bibinfo{author}{Fabrizio \surnamestart Montesi\surnameend}
  (\bibinfo{year}{2015}): \emph{\bibinfo{title}{Applied Choreographies}}.
\newblock {\sl \bibinfo{journal}{CoRR}} \bibinfo{volume}{abs/1510.03637}.
\newblock \urlprefix\url{http://arxiv.org/abs/1510.03637}.

\bibitemdeclare{article}{G10}
\bibitem{G10}
\bibinfo{author}{Daniele \surnamestart Gorla\surnameend}
  (\bibinfo{year}{2010}): \emph{\bibinfo{title}{Towards a unified approach to
  encodability and separation results for process calculi}}.
\newblock {\sl \bibinfo{journal}{Inf.\ Comput.}}
  \bibinfo{volume}{208}(\bibinfo{number}{9}), pp. \bibinfo{pages}{1031--1053},
  \doi{10.1016/j.ic.2010.05.002}.

\bibitemdeclare{inproceedings}{HMBCY11}
\bibitem{HMBCY11}
\bibinfo{author}{Kohei \surnamestart Honda\surnameend},
  \bibinfo{author}{A.~\surnamestart Mukhamedov\surnameend},
  \bibinfo{author}{G.~\surnamestart Brown\surnameend}, \bibinfo{author}{T.-C.
  \surnamestart Chen\surnameend} \& \bibinfo{author}{Nobuko \surnamestart
  Yoshida\surnameend} (\bibinfo{year}{2011}): \emph{\bibinfo{title}{Scribbling
  Interactions with a Formal Foundation}}.
\newblock In \bibinfo{editor}{Raja \surnamestart Natarajan\surnameend} \&
  \bibinfo{editor}{Adegboyega~K. \surnamestart Ojo\surnameend}, editors: {\sl
  \bibinfo{booktitle}{ICDCIT}}, {\sl \bibinfo{series}{LNCS}}
  \bibinfo{volume}{6536}, \bibinfo{publisher}{Springer}, pp.
  \bibinfo{pages}{55--75}, \doi{10.1007/978-3-642-19056-8\_4}.

\bibitemdeclare{article}{HYC16}
\bibitem{HYC16}
\bibinfo{author}{Kohei \surnamestart Honda\surnameend}, \bibinfo{author}{Nobuko
  \surnamestart Yoshida\surnameend} \& \bibinfo{author}{Marco \surnamestart
  Carbone\surnameend} (\bibinfo{year}{2016}): \emph{\bibinfo{title}{{Multiparty
  Asynchronous Session Types}}}.
\newblock {\sl \bibinfo{journal}{J. {ACM}}}
  \bibinfo{volume}{63}(\bibinfo{number}{1}), pp. \bibinfo{pages}{9:1--9:67},
  \doi{10.1145/2827695}.

\bibitemdeclare{inproceedings}{LGMZ08}
\bibitem{LGMZ08}
\bibinfo{author}{Ivan \surnamestart Lanese\surnameend},
  \bibinfo{author}{Claudio \surnamestart Guidi\surnameend},
  \bibinfo{author}{Fabrizio \surnamestart Montesi\surnameend} \&
  \bibinfo{author}{Gianluigi \surnamestart Zavattaro\surnameend}
  (\bibinfo{year}{2008}): \emph{\bibinfo{title}{Bridging the Gap between
  Interaction- and Process-Oriented Choreographies}}.
\newblock In \bibinfo{editor}{Antonio \surnamestart Cerone\surnameend} \&
  \bibinfo{editor}{Stefan \surnamestart Gruner\surnameend}, editors: {\sl
  \bibinfo{booktitle}{SEFM}}, \bibinfo{publisher}{{IEEE}}, pp.
  \bibinfo{pages}{323--332}, \doi{10.1109/SEFM.2008.11}.

\bibitemdeclare{inproceedings}{LNN16}
\bibitem{LNN16}
\bibinfo{author}{Hugo~A. \surnamestart L{\'{o}}pez\surnameend},
  \bibinfo{author}{Flemming \surnamestart Nielson\surnameend} \&
  \bibinfo{author}{Hanne~Riis \surnamestart Nielson\surnameend}
  (\bibinfo{year}{2016}): \emph{\bibinfo{title}{Enforcing Availability in
  Failure-Aware Communicating Systems}}.
\newblock In \bibinfo{editor}{Albert} \& \bibinfo{editor}{Lanese}
  \cite{forte2016}, pp. \bibinfo{pages}{195--211},
  \doi{10.1007/978-3-319-39570-8\_13}.

\bibitemdeclare{phdthesis}{M13:phd}
\bibitem{M13:phd}
\bibinfo{author}{Fabrizio \surnamestart Montesi\surnameend}
  (\bibinfo{year}{2013}): \emph{\bibinfo{title}{Choreographic {P}rogramming}}.
\newblock \bibinfo{type}{Ph.{D}. {T}hesis}, \bibinfo{school}{IT {U}niversity of
  {C}openhagen}.
\newblock
  \bibinfo{note}{\href{http://fabriziomontesi.com/files/choreographic_programming.pdf}{
  http://fabriziomontesi.com/files/choreographic\_programming.pdf }}.

\bibitemdeclare{inproceedings}{MY13}
\bibitem{MY13}
\bibinfo{author}{Fabrizio \surnamestart Montesi\surnameend} \&
  \bibinfo{author}{Nobuko \surnamestart Yoshida\surnameend}
  (\bibinfo{year}{2013}): \emph{\bibinfo{title}{Compositional Choreographies}}.
\newblock In \bibinfo{editor}{Pedro~R. \surnamestart D'Argenio\surnameend} \&
  \bibinfo{editor}{Hern{\'a}n~C. \surnamestart Melgratti\surnameend}, editors:
  {\sl \bibinfo{booktitle}{CONCUR}}, {\sl \bibinfo{series}{LNCS}}
  \bibinfo{volume}{8052}, \bibinfo{publisher}{Springer}, pp.
  \bibinfo{pages}{425--439}, \doi{10.1007/978-3-642-40184-8\_30}.

\bibitemdeclare{misc}{pi4soa}
\bibitem{pi4soa}
 (\bibinfo{year}{2008}): \emph{\bibinfo{title}{PI4SOA}}.
\newblock \bibinfo{howpublished}{\url{http://www.pi4soa.org}}.

\bibitemdeclare{inproceedings}{QZCY07}
\bibitem{QZCY07}
\bibinfo{author}{Zongyan \surnamestart Qiu\surnameend},
  \bibinfo{author}{Xiangpeng \surnamestart Zhao\surnameend},
  \bibinfo{author}{Chao \surnamestart Cai\surnameend} \&
  \bibinfo{author}{Hongli \surnamestart Yang\surnameend}
  (\bibinfo{year}{2007}): \emph{\bibinfo{title}{Towards the theoretical
  foundation of choreography}}.
\newblock In \bibinfo{editor}{Carey~L. \surnamestart Williamson\surnameend},
  \bibinfo{editor}{Mary~Ellen \surnamestart Zurko\surnameend},
  \bibinfo{editor}{Peter~F. \surnamestart Patel{-}Schneider\surnameend} \&
  \bibinfo{editor}{Prashant~J. \surnamestart Shenoy\surnameend}, editors: {\sl
  \bibinfo{booktitle}{WWW}}, \bibinfo{publisher}{ACM}, pp.
  \bibinfo{pages}{973--982}, \doi{10.1145/1242572.1242704}.

\bibitemdeclare{misc}{savara:website}
\bibitem{savara:website}
\bibinfo{author}{\surnamestart Savara\surnameend}: \emph{\bibinfo{title}{{JBoss
  Community}}}.
\newblock \bibinfo{note}{\url{http://www.jboss.org/savara/}}.

\bibitemdeclare{misc}{wscdl}
\bibitem{wscdl}
\bibinfo{author}{\surnamestart {W3C WS-CDL Working Group}\surnameend}
  (\bibinfo{year}{2004}): \emph{\bibinfo{title}{Web Services Choreography
  Description Language Version 1.0}}.
\newblock
  \bibinfo{howpublished}{\url{http://www.w3.org/TR/2004/WD-ws-cdl-10-20040427/}}.

\end{thebibliography}
